\documentclass{llncs}
\usepackage{amsmath}
\usepackage{amsfonts}
\usepackage{amssymb}

\spnewtheorem{heuristic}{Heuristic}{\itshape}{\rmfamily}

\usepackage{mathtools}
\usepackage{thmtools}

\newif\ifeprint
\eprinttrue

\usepackage{algorithm}
\usepackage{algorithmicx}
\usepackage{algpseudocode}
\usepackage[utf8]{inputenc}
\usepackage{tikz}
\usetikzlibrary{trees}
\usetikzlibrary{patterns}

\usepackage{pgfplots}
\usepackage{bbm}
\usepackage{multirow}
\usepackage{enumerate}
\usepackage{enumitem}
\usepackage{booktabs}
\usepackage{xcolor}
\usepackage{xifthen}

\usepackage{ dsfont }

\ifeprint % eprint version

\else % submission version

\fi

\pgfplotsset{compat=1.10}

\usepgfplotslibrary{fillbetween}
\colorlet{darkgreen}{green!50!black}
\colorlet{darkblue}{blue!50!black}
\colorlet{darkorange}{orange!90!black}
\colorlet{darkred}{red!50!black}

\usepackage{xifthen}

\newcommand{\bigO}[1]{ \ifthenelse{\isempty{#1}}{\mathcal{O}}{ \mathcal{O} \! \left(#1 \right)} }
\newcommand{\bigOt}[1]{ \ifthenelse{\isempty{#1}}{\widetilde{\mathcal{O}}}{ \widetilde{\mathcal{O}} \! \left(#1 \right)} }

\newcommand{\zo}{\{0,1\}}

\newcommand{\ket}[1]{\left| #1 \right\rangle}

\newcommand{\prob}[2][]{\Pr
\ifthenelse{\isempty{#1}}{}{_{#1}}
\ifthenelse{\isempty{#2}}{}{\left[ #2 \right]}
}

\usepackage[colorlinks=true]{hyperref}
\usepackage{cleveref}

\crefname{problem}{problem}{problems}
\Crefname{problem}{Problem}{Problems}
\crefname{proposition}{proposition}{propositions}
\Crefname{proposition}{Proposition}{Propositions}
\crefname{heuristic}{heuristic}{heuristics}
\Crefname{heuristic}{Heuristic}{Heuristics}

\renewcommand{\autoref}[1]{\Cref{#1}}

\def\epsilon{\varepsilon}

\algblockdefx[Repeat]{Repeat}{EndRepeat}%
[1]{\textbf{Repeat} #1 \textbf{times}}%
{\textbf{EndRepeat}}

\usepackage{xstring}
\makeatletter\def\setkey#1#2#3{\protected@edef#1##1{\protect\IfEq{#2}{##1}{#3}{#1{##1}}}}\makeatother
\def\newdict#1#2{\def#1##1{#2}}

\title{A Tight Quantum Algorithm for Multiple Collision Search}

\ifeprint
\author{Xavier Bonnetain\inst{1} \and Johanna Loyer\inst{2} \and André Schrottenloher\inst{3} \and Yixin Shen\inst{3}}
\institute{Université de Lorraine, CNRS, Inria, Nancy, France \and
Inria Saclay \and
Univ Rennes, Inria, CNRS, IRISA, Rennes, France}
\else
\author{}
\institute{}
\fi
\begin{document}
\maketitle
\renewcommand{\labelitemi}{$\bullet$}

\begin{abstract}
Searching for collisions in random functions is a fundamental computational problem, with many applications in symmetric and asymmetric cryptanalysis. When one searches for a \emph{single} collision, the known quantum algorithms match the query lower bound. This is not the case for the problem of finding \emph{multiple} collisions, despite its regular appearance as a sub-component in sieving-type algorithms.

At EUROCRYPT 2019, Liu and Zhandry gave a query lower bound $\Omega{}(2^{m/3 + 2k/3})$ for finding $2^k$ collisions in a random function with $m$-bit output. At EUROCRYPT 2023, Bonnetain et al. gave a quantum algorithm matching this bound for a large range of $m$ and $k$, but not all admissible values. Like many previous collision-finding algorithms, theirs is based on the MNRS quantum walk framework, but it \emph{chains} the walks by reusing the state after outputting a collision.

In this paper, we give a new algorithm that tackles the remaining non-optimal range, closing the problem. Our algorithm is tight (up to a polynomial factor) in queries, and also in time under a quantum RAM assumption. The idea is to extend the chained walk to a regime in which several collisions are returned at each step, and the ``walks'' themselves only perform a single diffusion layer.
\end{abstract}


\section{Introduction}

The collision search problem asks to find pairs of (distinct) inputs of a function that have the same output. It is a fundamental problem in quantum algorithms and cryptography, arguably only second to unstructured search. Many applications require to find many such collisions, which we can formulate as follows.

\begin{problem}[Multiple collision search]\label{pb:multicoll}
Given access to a (random) function $f~: \zo^n \to \zo^m$, where $n \leq m \leq 2n$, given $k \leq 2n-m$, find $2^k$ collisions of $f$, i.e., pairs of distinct $x,y$ such that $f(x) = f(y)$.
\end{problem}

The constraints on the parameters $n, m, k$ ensure that such collisions exist and that we are not asking for more collisions than the expected number of them ($\Theta(2^{2n-m})$ for a random function $f$ of this form).

\paragraph{Tightness of Existing Algorithms.}
Classically, it is well known that one can find $2^k$ collisions in $\bigO{}(2^{(m+k)/2})$ time and queries to $f$ (which is also a lower bound), against $\bigO{}(2^{m/2})$ for a single collision. Indeed, a list of $2^{(m+k)/2}$ random queries to $f$ can be expected to contain (the order of) $2^{(m+k) - m} = 2^k$ collisions. This simple algorithm is tight.

In the quantum setting, Aaronson and Shi~\cite{DBLP:journals/jacm/AaronsonS04} showed a query lower bound $\Omega(2^{m/3})$ for a single collision, which is matched by the BHT~\cite{DBLP:conf/latin/BrassardHT98} and Ambainis'~\cite{DBLP:journals/siamcomp/Ambainis07} algorithms for different values of $m$, thereby solving the problem for $2^k = 1$.

For a general $k$, Liu and Zhandry~\cite{DBLP:conf/eurocrypt/LiuZ19} showed a lower bound $\Omega(2^{2k/3+m/3})$. It is common knowledge that by extending the BHT algorithm to multiple outputs, one can reach the bound for $k \leq 3n-2m$, in queries to $f$, and also in time if one assumes quantum-accessible classical memory. The \emph{chained quantum walk} introduced in~\cite{DBLP:conf/eurocrypt/BonnetainCSS23} reaches the same complexity for $k \leq \min(2n-m, m/4)$, in queries, and in time if one assumes quantum-accessible quantum memory. Unfortunately, the union of these two parameter ranges does not entirely cover all $n,k,m$ such that $n \leq m \leq 2n$ and $k \leq 2n-m$. A range of $(k,m)$, forming a triangle contained in $[n/3;n] \times [n; 1.6n]$, remains only covered by basic generalizations of both algorithms which reach a suboptimal query and time complexity.

%label names, in a key-value dictionary
\newdict{\xlabel}{n}
\setkey{\xlabel}{1}{4n/3}
\setkey{\xlabel}{2}{3n/2}
\setkey{\xlabel}{3}{2n}

\newdict{\ylabel}{0}
\setkey{\ylabel}{1}{n/4}
\setkey{\ylabel}{2}{n/3}
\setkey{\ylabel}{3}{2n/5}
\setkey{\ylabel}{4}{2n/3}
\setkey{\ylabel}{5}{n}

\begin{figure}[htbp]
\centering
\begin{tikzpicture}
\begin{axis}[
scale=1., legend pos=outer north east,
xlabel={$m$}, ylabel={$k$},
xtick={1,1.333,1.5,2},
ytick={0,0.25,0.333,0.4,0.667,1},
xticklabel={$\xlabel{\ticknum}$},
yticklabel={$\ylabel{\ticknum}$},
ymin=0,ymax=1.0,xmin=1,xmax=2.0, xmajorgrids,ymajorgrids,
title={},
legend style={cells={anchor=west},name=legend,at={(1.1,0.5)},anchor=west}
]
\addplot[name path=limit] coordinates { (1,1) (2,0)};
\addplot[name path=bht] coordinates { (1,1) (1.5,0) };
\addplot[name path=nous] coordinates { (1,0.25) (8/5,8/20) (2,0) };
\addplot[name=vertical] coordinates { (1.333333, 0.333333) (1.333333, 0.666666) };
\addplot[name path=tmp] coordinates {(0,0) (2,0)};
\path[name path=axis] (0,0) -- (2,0);
\addplot [ thick, color=blue, fill=blue,  fill opacity=0.2]
    fill between[ of= bht and tmp];
\addplot [ thick, color=red, fill=red,  fill opacity=0.2 ]
    fill between[ of= nous and tmp];
\addplot [ thick, color=green, fill=green,  fill opacity=0.2 ]
    fill between[ of= limit and bht, soft clip={domain=1:1.333333}];
\addplot [ thick, color=orange, fill=orange,  fill opacity=0.2 ]
    fill between[ of= limit and nous, soft clip={domain=1.333333:2}];
\node[color=darkblue] at (axis cs:1.13,0.5) {\begin{tabular}{c}BHT \\ $\frac{2k}{3} + \frac{m}{3}$ \end{tabular}};
\node[color=darkred] at (axis cs:1.7,0.15) {\begin{tabular}{c} \cite{DBLP:conf/eurocrypt/BonnetainCSS23} \\ $\frac{2k}{3} + \frac{m}{3}$ \end{tabular}};
\node[color=darkorange] at (axis cs:1.7,0.6) {\begin{tabular}{c} \cite{DBLP:conf/eurocrypt/BonnetainCSS23} (extended) \\ $k + \frac{m}{4}$ \end{tabular}};
\node[color=darkgreen] at (axis cs:1.4,0.8) {\begin{tabular}{c}BHT (extended) \\ $k + m - n$ \end{tabular}};
\end{axis}
\end{tikzpicture}
\caption{Parameter ranges and complexity exponents of quantum multiple collision search algorithms. This figure is taken from~\cite{DBLP:conf/eurocrypt/BonnetainCSS23}. (All complexities are both in queries, and gate count, assuming qRAM gates).}
\label{fig:parameter-ranges}
\end{figure}
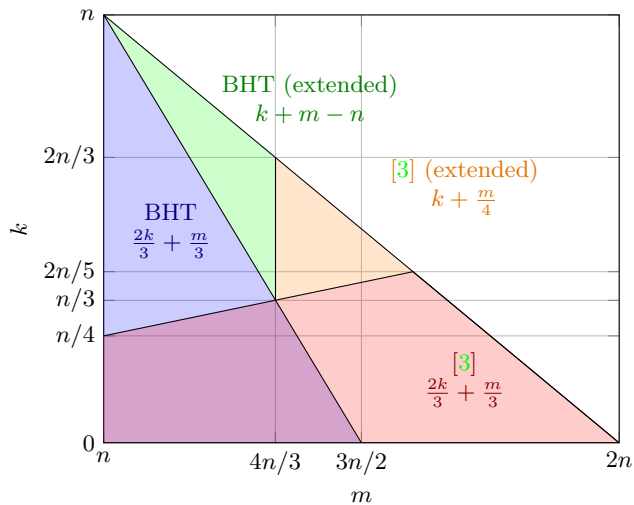

\paragraph{Contribution.}
In this paper, we give a new algorithm that covers this remaining range. Like in~\cite{DBLP:conf/eurocrypt/BonnetainCSS23}, we use a \emph{chained} version of the MNRS quantum walk framework~\cite{DBLP:journals/siamcomp/MagniezNRS11}. The idea of the chained quantum walk is to reuse the final state of a walk as starting state for a new one. Indeed, the state of a walk is a superposition of subsets of $\{0,1\}^n$. At the end of a walk, one expects the subsets to always contain a collision: one can \emph{extract} and measure this collision, collapsing the vertex on a superposition of subsets which is still close to uniform -- and sufficient to carry on.

In this new version, we increase the vertex size so much that in fact, the quantum state always contains (exponentially many) collisions on average. Now, there is no need to properly \emph{walk} during the quantum walks, and they are reduced to a single diffusion layer. The effect of the diffusion is to re-randomize the vertex, making new collisions appear.

This allows us to prove the following:

\begin{restatable}{theorem}{mainthm}
\label{thm:main}
For all $n \leq m \leq 2n$, $k \leq 2n-m$, there exists a quantum algorithm solving the multiple collision search problem in $\bigOt{2^{2k/3 + m/3}}$ queries, memory and total gate count (assuming qRAM gates). Furthermore, for any $\ell \leq 2k/3 + m/3$, there exists a quantum algorithm using memory $2^{\ell}$ and  $\bigOt{2^{k + m/2 -\ell/2}}$ queries and total gate count (assuming qRAM gates).
\end{restatable}

With the lower bound from \cite{DBLP:conf/eurocrypt/LiuZ19} this implies, as a direct corollary:
\begin{corollary}[Complexity of Multiple collision search]
 The query complexity of \autoref{pb:multicoll} is, up to a factor polynomial in $n$, $2^{m/3+2k/3}$.
\end{corollary}

% \mainthm*

\paragraph{Applications in Symmetric Cryptography.}
In symmetric cryptography, the \emph{limited birthday} problem asks, given black-box oracle access to a permutation $E$ of $\{0,1\}^n$ (a block cipher), to find many pairs $P, P'$ such that $E(P) \oplus E(P')$ and $P \oplus P'$ both belong to a given vector subspace of $\{0,1\}^n$. It is an important part of impossible differential attacks~\cite{DBLP:conf/asiacrypt/BouraNS14}, and also serves as a distinguisher of random permutations~\cite{DBLP:journals/tosc/HosoyamadaNS20}. The limited birthday problem is solved by restricting the inputs of $E$ (or outputs) to \emph{structures} and solving a multiple collision search problem on such structures. Depending on the sizes of the restricting vector subspaces, one can encounter both cases with \emph{few} collisions to output, or all of them, making this a very versatile application~\cite{DBLP:journals/dcc/DavidNS24}.

\paragraph{Applications in Asymmetric Cryptography.}
Sieving algorithms are currently the fastest at solving the shortest vector problem on integer lattices, and are used to estimate the security of post-quantum lattice-based cryptosystems. They also exist in the cryptanalysis of code-based schemes~\cite{DBLP:conf/eurocrypt/DucasEEK24,DBLP:journals/dcc/EngelbertsEL25}. Sieving algorithms construct several lists of elements (e.g., vectors) where the next ones are obtained by finding partial collisions between the previous ones. Previously, the generic improvement of multiple collision-finding in~\cite{DBLP:conf/eurocrypt/BonnetainCSS23} allowed to improve (slightly) the previous best quantum algorithm for lattice sieving given in~\cite{CL21}. In this paper, we do not improve lattice sieving (since we only improve multiple-collision search in an adjacent regime), but we believe that our new technique could likely be used inside similar algorithms.

\paragraph{Finding multicollisions.} The problem of multicollision-finding in the quantum setting has only been studied in the case $n = m$~\cite{DBLP:conf/eurocrypt/LiuZ19}. As the known algorithms in this case start by computing a large number of collisions, a tight algorithm to find many collisions is likely a first step towards a tight quantum algorithm for multicollision-finding in the case $m > n$.

% !TeX root = ../triangle-doom.tex

\section{Preliminaries}

We assume knowledge of the basics of quantum computing in the circuit model~\cite{nielsen2002quantum} such as qubits, quantum states, unitaries, \emph{etc.} In order to solve~\autoref{pb:multicoll} in the quantum setting, the function $f$ is accessed as a unitary oracle:
$$ \ket{x} \ket{y} \xmapsto{O_f} \ket{x} \ket{y \oplus f(x)} \enspace, $$
where $x$ is an $n$-bit input and $y$ an $m$-bit input.

We estimate the complexities of our algorithms asymptotically in $n,m$ and $k$, as follows. The \emph{query complexity} is the number of queries to $O_f$ (which also includes classical queries to $f$ as a special case). The \emph{time complexity} (or gate count) is the number of elementary quantum gates; the \emph{memory complexity} is the number of qubits or classical bits of memory. Since we are only interested in asymptotic complexities, any set of universal quantum gates is sufficient; one can take the Clifford+T gate set.

In order to reach optimal time complexities, we need \emph{quantum RAM} (qRAM). Different variants of qRAM are used in collision search algorithms: for example, the BHT algorithm requires only \emph{quantum-accessible classical memory}~\cite{DBLP:conf/latin/BrassardHT98}. In our case, we need \emph{quantum-accessible quantum memory}. Like previous works~\cite{DBLP:journals/siamcomp/Ambainis07}, we formalize our qRAM assumption by introducing the so-called qRAM gate:
$$ \ket{y_1, \ldots, y_M}  \ket{x} \ket{i} \xmapsto{\mathsf{qRAM}} \ket{y_1, \ldots, y_{i-1}, x, y_{i+1}, \ldots, y_M} \ket{y_i} \ket{i} \enspace. $$
We give cost 1 to the qRAM gate (like other elementary gates) even if $M$ is exponential-size; the qRAM gate can even span the entire quantum circuit. While this is obviously a very powerful model, it has no incidence on the query complexity, since a qRAM gate can always be simulated by an exponential-size Clifford+T circuit.

\subsection{Quantum Search}

%\AS{TODO: check that this is not copied from the prev paper (I don't remember)}

Let $\ket{\psi_U}$ be a quantum state (the ``uniform superposition'') produced by some \textbf{Setup} algorithm, which is a superposition of two orthogonal components $\ket{\psi_G}$ (the ``good subspace'') and $\ket{\psi_B}$ (the ``bad subspace''):
$$ \ket{\psi_U} = \beta \ket{\psi_B} + \alpha \ket{\psi_G} $$
for some parameters $\alpha, \beta$. Let $Ref_{\ket{\psi_B}}$ be the reflection over $\ket{\psi_B}$, that performs:
$$ Ref_{\ket{\psi_B}} \ket{\psi_B} = \ket{\psi_B} , \quad Ref_{\ket{\psi_B}} \ket{\psi_G} = - \ket{\psi_G} \enspace.$$
We assume that this reflection can be performed by an oracle $O_f$; that is, there is a Boolean function $f$ defined on the outputs of \textbf{Setup} that returns 1 iff an output belongs to the good subspace. We can either consider a \emph{phase oracle} for $f$ that immediately implements $Ref_{\ket{\psi_B}}$, or a computational oracle.

%One iterate of quantum walk corresponds to one iterate of Grover search, whose sole purpose is to ``switch back'' our current bad state to a good state, without having to reconstruct a good state from scratch.

Grover's algorithm~\cite{DBLP:conf/stoc/Grover96} and more generally amplitude amplification~\cite{brassard2002quantum}, consist in alternating reflections over $\ket{\psi_B}$, which are implemented using a \emph{checking oracle} that recognizes the components of $\ket{\psi_G}$, and reflections over $\ket{\psi_U}$, which can be implemented by calling the \textbf{Setup} algorithm. Indeed, it can be shown that $Ref_{\ket{\psi_U}} Ref_{\ket{\psi_B}}$ is a rotation of angle $2\arcsin \alpha$ in the plane spanned by $\ket{\psi_B}, \ket{\psi_G}$.

%After $k$ iterations of $Ref_{\ket{\psi_U}} Ref_{\ket{\psi_B}}$, the state becomes:
%\[ \cos((2k+1)\arcsin \alpha) \ket{\psi_B} + \sin((2k+1) \arcsin \alpha) \ket{\psi_G} \]

We view amplitude amplification as a way to transform a linear combination of $\ket{\psi_B}, \ket{\psi_G}$ into another. 
%If we start from a \emph{typical} state, we have the following.

\begin{lemma}\label{lemma:qaa}
Assume that $\ket{\psi_U} = \beta \ket{\psi_B} + \alpha \ket{\psi_G}$. Suppose that $\alpha \leq \frac{1}{\sqrt{2}}$. 
Starting from $\ket{\psi_U}$ or $\ket{\psi_B}$, there exists $ k =\bigO{1 / \alpha}$ such that, after $k$ iterations of $Ref_{\ket{\psi_U}} Ref_{\ket{\psi_B}}$, the resulting state is a linear combination $\beta' \ket{\psi_B} + \alpha' \ket{\psi_G}$ where $\alpha' = \bigO{1}$.
\end{lemma}

\begin{proof}
We have:
$$ \ket{\psi_U} = \sin \theta \ket{\psi_G} + \cos \theta \ket{\psi_B} $$
where $\theta = \arcsin \alpha$, so after $k$ iterations we obtain:
\[ \cos( (\arcsin \alpha) + 2k (\arcsin \alpha) ) \ket{\psi_B} + \sin( (\arcsin \alpha) + 2k(\arcsin \alpha) ) \ket{\psi_G} \enspace. \]
We have:
\begin{align*}
\sin( (\arcsin \alpha) + 2k(\arcsin \alpha) ) &\geq (2k+1) \arcsin \alpha - ((2k+1) \arcsin \alpha )^2 / 6 \\
& \geq (2k+1) \arcsin \alpha - ((2k+1) \alpha \pi )^2 / 24 \enspace.
\end{align*}
This can be made a constant by choosing $k = \bigO{1 / \alpha}$. If we start from $\ket{\psi_B}$, the initial angle is 0 instead, but the asymptotic behavior is the same. \qed
\end{proof}

One can be more precise with the number of iterations required to achieve $\alpha'$ close to 1, but this is not necessary for us since we focus on asymptotic complexities only.

%We are also interested in the case where we start from a non-typical superposition.
%Next, if we start from $\ket{\psi_B}$, we have the following.

\begin{lemma}\label{lemma:qaa-bad}
Suppose that $\alpha \leq \frac{1}{\sqrt{2}}$. Starting from $\ket{\psi_G}$, there exists $k = \bigO{1 / \alpha}$ such that after $k$ iterations of $Ref_{\ket{\psi_U}} Ref_{\ket{\psi_B}}$, the resulting state is a linear combination $\beta' \ket{\psi_B} + \alpha' \ket{\psi_G}$ where $\beta' = \bigO{1}$.
\end{lemma}

\begin{proof}
The initial angle in the plane spanned by $\ket{\psi_B}, \ket{\psi_G}$ is $\frac{\pi}{2}$ in that case, so after $k$ iterations we have:
\begin{multline*}
 \cos( \pi/2 + 2k (\arcsin \alpha) ) \ket{\psi_B} + \sin(\pi/2 + 2k(\arcsin \alpha) ) \ket{\psi_G} \\
 = \sin( 2k (\arcsin \alpha) ) \ket{\psi_B} - \cos( 2k(\arcsin \alpha) ) \ket{\psi_G} \enspace.
 \end{multline*}
Thus, by choosing $k = \bigO{1/\alpha}$ we get a constant amplitude on $\ket{\psi_B}$.\qed
\end{proof}

In both cases, we navigate in the plane spanned by $\ket{\psi_G}, \ket{\psi_B}$. If we assume that one of them (their roles can be exchanged) represents a ``typical'' case and the other an ``atypical'' case (with an amplitude $\alpha$ in the uniform superposition), then \autoref{lemma:qaa} shows that we can move from the typical case to the atypical one with $\bigO{1/\alpha}$ iterations. Conversely, \autoref{lemma:qaa-bad} shows that one can move from the atypical case back to the typical one with $\bigO{1/\alpha}$ iterations.

%When the two cases have constant amplitude, we can go back and forth between them with an expected constant number of iterations.

%We will also need the case of a quantum search with a \emph{single iteration}.

After performing the iterations, we will project the result and see if we obtain $\ket{\psi_B}$ or $\ket{\psi_G}$. If we did not obtain the wanted state, we can just restart. The average number of repetitions will be constant. The result is summarized in~\autoref{cor:flip}, and the procedures are given in~\autoref{alg:search}, which allows to transform $\ket{\psi_B}$ into $\ket{\psi_G}$, \emph{and the converse}, at the same cost.

%\begin{lemma}\label{lemma:flip}
%Assume that $\ket{\psi_U} = \beta \ket{\psi_B} + \alpha \ket{\psi_G}$ where $\alpha, \beta = \bigO{1}$ are constants. Starting from $\ket{\psi_U}$ or $\ket{\psi_B}$, apply $Ref_{\ket{\psi_U}} Ref_{\ket{\psi_B}}$ once, then apply $f$ and measure. Then there exists $\gamma = \bigO{1}$ such that:
%\begin{itemize}
%\item With probability $\gamma$, one measures 1 and the state has collapsed to $\ket{\psi_G}$
%\item With probability $1-\gamma$, one measures 0 and the state has collapsed to $\ket{\psi_B}$
%\end{itemize}
%\end{lemma}

%As a corollary, we can transform $\ket{\psi_B}$ into $\ket{\psi_G}$ with a constant expected number of steps. This procedure repeatedly applies $Ref_{\ket{\psi_U}} Ref_{\ket{\psi_B}}$ and measures. It also works more generally to transform $\ket{\psi_U}$ into $\ket{\psi_G}$ in $\bigO{1/\alpha}$ steps on expectation.

\begin{corollary}\label{cor:flip}
Assume that $\ket{\psi_U} = \beta \ket{\psi_B} + \alpha \ket{\psi_G}$ where $\alpha$ is at most a constant. There are two quantum procedures using an expected $\bigO{1/\alpha}$ calls to $ (Ref_{\ket{\psi_U}} Ref_{\ket{\psi_B}})$ that respectively:
\begin{itemize}
\item On input $\ket{\psi_B}$ or $\ket{\psi_U}$, return $\ket{\psi_G}$;
\item On input $\ket{\psi_G}$, return $\ket{\psi_B}$.
\end{itemize}
% There is a quantum procedure that, on input $\ket{\psi_B}$ or $\ket{\psi_U}$, returns $\ket{\psi_G}$, and runs in expected $\bigO{1/\alpha}$ steps. There is another procedure that on input $\ket{\psi_G}$, returns $\ket{\psi_B}$, and runs in expected expected $\bigO{1/\alpha}$ steps.
% where $\alpha, \beta = \bigO{1}$ are constants. There is a quantum procedure that, on input $\ket{\psi_B} $, returns $\ket{\psi_G}$ (with probability 1), and which contains on expectation a constant number of calls to $Ref_{\ket{\psi_U}}$ and $Ref_{\ket{\psi_B}}$.
\end{corollary}

\begin{algorithm}[htbp]
\caption{Procedure $\mathsf{Search}(\ket{\psi_B}, \ket{\psi_G}, \ket{\psi_U})$.}\label{alg:search}

\begin{algorithmic}[1]
\Statex Define: $\ket{\psi_U} = \beta \ket{\psi_B} + \alpha \ket{\psi_G}$ where $\ket{\psi_B}$ and $\ket{\psi_G}$ are orthogonal; an oracle $f$ to recognize states in $\ket{\psi_G}$
\medskip
\Statex \textbf{Input:} $\ket{\psi} \leftarrow \ket{\psi_B}$ or $\ket{\psi_U}$ (resp. $\ket{\psi_G}$)
\Statex \textbf{Output:} $\ket{\psi} \leftarrow\ket{\psi_G}$ (resp. $\ket{\psi_B}$)
\medskip
\State \textsf{Success} $\leftarrow$ False
\While{Not \textsf{Success}}
\State $\ket{\psi} \leftarrow (Ref_{\ket{\psi_U}} Ref_{\ket{\psi_B}})^{1/\alpha} \ket{\psi}$
\State Compute $f$: $\ket{\psi} \leftarrow \ket{\psi}\ket{f(\psi)}$
\State Measure the last bit (projects on $\ket{\psi_B}$ or $\ket{\psi_G}$)
\State If we obtained the wanted state, \textsf{Success} $\leftarrow$ True
\EndWhile
\State \textbf{Return} $\ket{\psi}$
\end{algorithmic}

\end{algorithm}

\subsection{MNRS Quantum Walks}

The MNRS quantum walk framework~\cite{DBLP:journals/siamcomp/MagniezNRS11} can be seen as an extension of quantum amplitude amplification where the implementation of $Ref_{\ket{\psi_U}}$ differs from the \textbf{Setup} algorithm.  We recall here some essential features, but for a more comprehensive account, one can refer to~\cite{DBLP:conf/eurocrypt/BonnetainCSS23}.

One considers a regular undirected graph $G = (V,E)$ with degree $d$, vertex set $V$ and edge set $E$, where a proportion $\epsilon$ of the vertices are \emph{marked} (noted $M \subseteq V$). The goal of the quantum walk is to recover a marked vertex. Assume that an encoding of vertices as orthogonal quantum states is given, denoted as $\ket{\hat{x}}$ for $x \in V$. For $x \in V$, let $N_x$ be the set of its neighbors. The MNRS quantum walk is analogous to a random walk on \emph{edges} of the graph $G$. One defines:
\begin{align*}
\ket{p_x} &:= \frac{1}{\sqrt{d}} \sum_{y \in N_x} \ket{\hat{y}} \\
\ket{\psi_U} &:= \frac{1}{\sqrt{V}} \sum_{x \in V} \ket{\hat{x}} \ket{p_x}, \quad & \ket{\psi_M} & := \frac{1}{\sqrt{M}} \sum_{x \in M} \ket{\hat{x}} \ket{p_x} \\
A & := \mathrm{span} ( \ket{\hat{x}} \ket{p_x} )_{x \in V}, \quad & B & := \mathrm{span} ( \ket{p_x} \ket{\hat{x}} )_{x \in V}
\end{align*}

The MNRS framework emulates $Ref_{\ket{\psi_U}}$ by a \emph{phase estimation} of the unitary $Ref_B Ref_A$. This phase estimation requires $\bigO{}(1/\sqrt{\delta})$ calls to both $Ref_B$ and $Ref_A$, where $\delta$ is the \emph{spectral gap} of the graph $G$. The spectral gap is related to classical random walks: $\bigO{}(1/\delta)$ is the number of random walk steps to take, starting from any vertex, before the current vertex becomes uniformly distributed.

Both $Ref_B$ and $Ref_A$ can be implemented from a unitary that, on input $\ket{\hat{x}}$, produces the superposition of neighbors $\ket{p_x}$. This is what is commonly known as the \textbf{Update} unitary. Finally, one needs a \textbf{Check} unitary that flips the phase of bad vertices. The algorithm is often constructed so that \textbf{Check} has a trivial implementation. By applying~\autoref{alg:search} on the output of \textbf{Setup}, we obtain the following.

\begin{theorem}[MNRS (informal)]
There is a quantum algorithm that, using 1 call to \textbf{Setup}, and on expectation $\bigO{\frac{1}{\sqrt{\epsilon}} \frac{1}{\sqrt{\delta}} }$ calls to \textbf{Update} and $\bigO{\frac{1}{\sqrt{\epsilon}} }$ calls to \textbf{Check}, returns $\ket{\psi_M}$.
\end{theorem}

While this point of view suffices for query complexity, the state of a vertex is a very large one. Thus, keeping two entire vertices in memory is wasteful, and having to reconstruct an entire vertex during the \textbf{Update} is too costly.

\paragraph{Johnson Graph and Radix Trees.}
Most MNRS quantum walks in cryptanalysis use \emph{Johnson graphs}, for which this problem can be solved generically: the vertex-vertex pair can be compressed to a single vertex and a \emph{coin} indicating the next one. This is detailed in~\cite{DBLP:conf/eurocrypt/BonnetainCSS23}.

A Johnson graph $J(X,R)$ is the graph where the vertices are the subsets of $X$ containing $R$ elements. Two sets are adjacent if they have exactly one differing element. In collision search and similar algorithms, the vertex $x$ is encoded into $\ket{\hat{x}}$ by using a specific set data structure, and possibly computing additional information. In our case we store pairs $(z,f(z))$ as $(f(z)| z)$ using a concatenation of bit-strings.

The quantum radix trees~\cite{DBLP:phd/basesearch/Jeffery14} allow both to represent efficiently an ordered set and to perform complex operations efficiently (assuming qRAM), such as adding an element, removing an element and finding an element. We can also augment the radix trees with useful information such as the number of collisions that each branch contains (in order to quickly find collisions, without having to write them down). All of this side information is part of the encoding $\ket{\hat{x}}$, and also makes trivial any implementation of \textbf{Check} that would rely on the number of collisions in a vertex.

\subsection{Statistics on the Number of Collisions}

In the remainder of this paper, we will work with \emph{multicollision tuples} rather than collision pairs. By multicollision, we mean a unique tuple of elements which all have the same image by $f$ (for example, a 3-collision is not counted as a pair of collisions). Our algorithm will return such tuples.

In a random function on $n$ bits, the largest size of a multicollision is expected to be $\bigO{n}$. This means that up to a polynomial factor, the number of multicollisions and the number of collisions returned by our algorithm is the same; which is why we can focus on the former. This is more practical and equivalent to~\autoref{pb:multicoll} up to a polynomial factor. The exceptional case (a function having a larger multicollision than expected by this bound) is taken as a case of failure of our algorithm.

%\JL{
\begin{lemma}[{\cite[Theorem~2]{DBLP:journals/ieicet/SuzukiTKT08}}]
Let be $n,m,\ell$ with $n \leq m$. Then the probability $P(n,m,\ell)$ that a random function $f:X \rightarrow Y$ where $|X|=2^n$ and $|Y|=2^m$ has a $\ell$-multicollision satisfies $P(n,m,\ell) \leq 2^{-m(\ell-1)} \cdot \binom{2^n}{\ell}.$
\end{lemma}
%}

\begin{corollary}
There exist constants $\epsilon$ and $c$ such that with probability $1 - \epsilon$, all multicollisions in a random function are tuples of size $\leq cn$. 
\end{corollary}

%\JL{
\begin{proof}
Let $\ell =cn$. By applying the previous lemma and the standard bound $\binom{2^n}{\ell} \leq (e 2^n/\ell)^\ell$, we get $P(n,m,\ell) \leq \left( 2^{-n(1-1/\ell)} \frac{e 2^n}{\ell} \right)^{cn}$ as $n \leq m$. For large $n$, the union bound over all $\ell \geq cn$ shows that a multicollision of size larger than $cn$ exists with negligible probability $\epsilon \leq 2^{-\Omega(n)}$. \qed
\end{proof}	
%}

Let us denote $X = \{0,1\}^n$ and $Y = \{0,1\}^m$. As our main algorithm considers mutiple walks, we will need to define reduced sets $X'$ and $Y'$ and a corresponding function $f'~: X' \to Y'$ which is simply a restriction of $f$.

We are interested in the number of multicollisions that can appear in a vertex. The vertex is a set $S'$ of elements taken uniformly at random from $X'$, with $|S'| < |X'|$. All three of $|S'|$, $|X'|$ and $|Y'|$ can be varying parameters, but we will ensure that they do not vary too much throughout the algorithm, so that our statistics remain valid from one walk to another.

We use the following heuristic.

\begin{heuristic}\label{heuristic: nb multicollisions}
The number of multicollisions in a set follows a Poisson distribution. When its expectation is large, the number of multicollisions follows a normal distribution.
\end{heuristic}

Under this heuristic, we obtain the following result, which is our main tool to handle the statistics of collisions in vertices.

\begin{lemma}\label{lemma:collisions}
Let $Z(S')$ be the number of multicollisions (of $f'$) in the set $S'$, with elements selected uniformly at random from $X'$. Assume that $R/2 \leq |S'| \leq R$ and $M/2 \leq |Y'| \leq M$ for some global parameters $R$ and $M$, where $8R < M$. There is a constant $c$ such that, for any pair of constants $c_1, c_2$, the probability: 
%\JL{$M'=|Y'|$ ?} \AS{Oui (le point étant que la valeur de la moyenne varie au cours de la marche, mais la valeur de l'intervalle peut rester la même et ça simplifie un peu)}
\begin{equation}
\prob[S']{ c \frac{|S'|^2}{|Y'|} - c_1 \frac{R}{\sqrt{M}}  \leq Z(S') \leq c \frac{|S'|^2}{|Y'|} + c_2 \frac{R}{\sqrt{M}}  }
\end{equation}
is also constant.
\end{lemma}

\begin{proof}
Let us write $Z := Z(S')$ for simplicity. 
By definition, we have $Z = \#\{ y \in Y' : N_y \geq 2\}$ with $N_y := \#\{x \in S' : f'(x)=y\}$. 
By \Cref{heuristic: nb multicollisions}, we can approximate $p := \Pr(N_y \geq 2) = 1 - e^{-\lambda} - \lambda e^{-\lambda}$ with mean $\lambda = |S'|/|Y'|$. 
%\JL{Je suppose ici $4|S'| \leq |Y'|$: Est-ce toujours le cas dans notre contexte ? Sinon on peut remplacer par d'autres constantes}\AS{ajouté}
For $\lambda \in (0,1/4)$, we have $\lambda^2/2 \leq 1- e^{-\lambda}(1+\lambda) \leq 2\lambda^2/3$. Note that $\lambda < 1/4$ follows from $8R < M$, and so $4|S'| < |Y'|$ for any values of $S'$ and $Y'$.
We have $\mathbb{E}[Z] = |Y'| \cdot p$, hence 
\begin{equation} \label{eq: bound mean}
	\frac{|S'|^2}{2|Y'|} \leq \mathbb{E}[Z] \leq \frac{2|S'|^2}{3|Y'|}.
\end{equation}
	
Now we aim to bound the variance. 
For distinct $y,z \in Y'$, the events $(N_y \geq 2)$ and $(N_z \geq 2)$ are negatively correlated, since having more elements $x \in S'$ with images $f'(x)=y$ decreases the chance that another $z \in Y'$ admits at least two preimages in $S'$. So the covariance of these two events is negative, and thus
	$$\mathrm{Var}(Z) = \sum_{y} \mathrm{Var}(\mathds{1}_{(N_y\geq 2)}) + 2 \sum_{y<z} \mathrm{Cov}(\mathds{1}_{(N_y\geq 2)}, \mathds{1}_{(N_z\geq 2)}) \leq \sum_y \mathrm{Var}(\mathds{1}_{(N_y \geq 2)})$$
We have $\mathrm{Var}(\mathds{1}_{(N_y \geq 2)}) = p(1-p) \leq p$, then $\mathrm{Var}(Z) \leq \mathbb{E}[Z]$. 
By \Cref{eq: bound mean}, the standard deviation then satisfies $\sigma_Z \leq \sqrt{2/3} \frac{|S'|}{\sqrt{|Y'|}} \leq C \frac{R}{\sqrt{M}}$ for constant $C>0$, as $|S'|\leq R$ and $|Y'|\geq M/2$. 

For any fixed constant $c>0$, we deduce
	$$\Pr \left( \frac{|S'|^2}{2|Y'|} - cC \frac{R}{\sqrt{M}} \leq Z \leq \frac{2 |S'|^2}{3|Y'|} + cC \frac{R}{\sqrt{M}} \right)  \leq \Pr(|Z - \mathbb{E}[Z]| \leq c\sigma_Z)$$
and by applying Chebyshev's inequality, we know this probability is at least $1 - 1/c^2$, which is constant.  
\qed
\end{proof}

When we take a random subset of size $R$ with codomain of size $M$, with many multicollisions, the average number of multicollisions is $E = \bigO{}(R^2 / M)$, and the standard deviation is $T = \bigO{}(R / \sqrt{M})$. As a consequence of~\autoref{lemma:collisions} we can say that a constant proportion of subsets have a number of multicollisions in the interval $[E; E+T]$; and also, a constant proportion fall in the interval $[E-T; E]$. Note that the technical condition $8R < M$ in~\autoref{lemma:collisions} is also obviously validated in our applications, as the vertex size is at most of order $\bigO{2^n}$, and the domain size at least $\bigO{2^n}$, so we can always leave a constant factor of difference between these two.

We can even go further. Because of the evolution of $R$ and $M$ throughout our algorithm, the average of the number of multicollisions will vary as well. But if we don't change $R$ and $M$ ``too much'', then we still have constant probability to fall into the intervals with the new value of $E'$.

\begin{lemma}\label{lemma:modif-e}
Fix $T =  \bigO{R / \sqrt{M}}$ and assume $R \ll M$. Let $R' \in [R - T; R]$ and $M' \in [M-T; M]$. Let $E' = c \frac{(R')^2}{M'}$ where $c$ is the universal constant of~\autoref{lemma:collisions}. Then the number of collisions in a random vertex of size $R$ falls in the intervals $[E'; E'+T]$ and $[E'-T; E']$ with constant probability.
\end{lemma}

\begin{proof}
We simply prove that the difference between $E'$ and $E$ is smaller than the standard deviation. Indeed, we have:
\begin{equation}
E' = \frac{(R')^2}{M'} = \frac{(R - x)^2}{M - y}
\end{equation}
for some $x,y \in [0; R / \sqrt{M}]$, which is much smaller than both $R$ and $M$. Thus:
\begin{align*}
E' &= \frac{R^2}{M} \left( 1 - \frac{x}{R} \right)^2  \left( 1 - \frac{y}{M} \right)^{-1} \\
& \simeq \frac{R^2}{M} \left( 1 - \frac{2x}{R} + \frac{y}{M} \right) \enspace.
\end{align*}
We notice then that: $ \frac{R^2}{M} \frac{x}{R} = \frac{Rx}{M} \leq \frac{R^2}{M^{3/2}} \ll \frac{R}{\sqrt{M}} $ as we have $R \ll M$. Besides: $\frac{R^2}{M} \frac{y}{M} \ll y = R / \sqrt{M}$.\qed
\end{proof}

In summary, if the parameters of the set are modified up to an amount equal to the standard deviation of the number of collisions; then our reasoning on intervals still holds.

%In summary, when   Furthermore a constant proportion of such subsets are in each of the following intervals:
%\begin{equation}
%[0; E-T] \quad [E-T; E] \quad [E; E+T] \quad [E+T, \infty] \enspace.
%\end{equation}
%Modifying the parameters $R$ and $M$ (up to a factor 2) influences the value of $E$, and this could be a problem if it deviates too much (e.g., by a constant factor) throughout our algorithm.
% However 
%, but \emph{we can keep the value of $T$ identical}, and we still have a global (constant) lower bound on the probability to fall in each of these intervals.

\subsection{Main Ideas}

In order to explain our main idea, we first briefly recall the algorithms of~\cite{DBLP:journals/siamcomp/Ambainis07} and~\cite{DBLP:conf/eurocrypt/BonnetainCSS23} (with slight modifications).

Ambainis' element distinctness algorithm~\cite{DBLP:journals/siamcomp/Ambainis07} can be reframed as an MNRS quantum walk on a Johnson graph $J(\{0,1\}^n, R)$, where vertices are subsets of $\{0,1\}^n$ of size $R$. A vertex is marked if it contains a collision. Therefore, the probability for a vertex to be marked is $\bigO{R^2 / 2^m}$. The spectral gap is $\delta = 1/R$. The cost of \textbf{Setup} is $\bigO{R}$, the cost of \textbf{Update} is negligible (using a quantum radix tree), and the total gate count of the algorithm is: $\bigOt{R + \frac{2^{m/2}}{R} \sqrt{R} }$ which is optimized by setting $R = 2^{m/3}$ and obtaining $\bigOt{2^{m/3}}$.

\paragraph{Chained Walk.}
In the \emph{chained} quantum walk, the goal is to obtain many collisions. One starts with a vertex of size $R$, and performs a first quantum walk. At the end of this walk, one obtains the uniform superposition of marked vertices, i.e., subsets of $\{0,1\}^n$ containing at least one (multi)collision. These (multi)collisions are \emph{removed} from the subset, and measured. Crucially, the rest of the vertex is not measured. The remaining elements form a uniform superposition over:
\begin{itemize}
\item Subsets of size $R-t$ where $t$ is the number of elements measured;
\item Excluding the elements of $\{0,1\}^n$ which were measured;
\item Excluding all other elements of $\{0,1\}^n$ which would have an image equal to the images measured;
\item Which \emph{do not} contain a (multi)collision. 
%\YS{si on extrait juste un multicollisions, le reste peut contenir des multicollisions aussi.}\AS{corrigé}
\end{itemize}
This is the \emph{extraction} step of the walk. After outputting this collision, the remaining state can be reused as the starting state in a new walk, with a restricted function $f$ and a reduced vertex size. Note that this starting state is actually a superposition of \emph{unmarked} vertices (that do not contain a multicollision), but this is enough for quantum search.

The gate count is of order:
\begin{equation}
R + 2^k \frac{2^{m/2}}{R} \sqrt{R} \enspace.
\end{equation}
%One can use any vertex size which satisfies the constraints:
%\begin{equation}
%2^k \leq R \leq \min ( 2^{\frac{m}{3} + \frac{2k}{3}} , 
%\end{equation}
One can use any vertex size $R$ as long as $R \leq 2^{m/2}$: afterwards, vertices contain on average one collision or more, and this formula is not valid anymore.

Therefore, in order to reach the optimal complexity $\bigOt{}(2^{\frac{m}{3} + \frac{2k}{3}})$ (obtained for $R = 2^{\frac{m}{3} + \frac{2k}{3}}$), one needs to have $k \leq \frac{m}{4}$. In other words, the algorithm fails when \emph{the number of collisions to output is too high} with respect to $m$, and is not optimal in a regime with intermediate $m$ (e.g, $m = 1.5 n$) where we would want to output all collisions.

\paragraph{New Regime.}
Our new algorithm follows the principle of the chained quantum walk (the final state of a walk is reused as the initial state of the next one), but fits in the regime where a vertex contains more than one (multi)collision on average.

As a simplified view, let us neglect the fact that extracting multicollisions post-selects the remaining elements, modifying their distribution (and the domain of the quantum walk). Indeed, as long as the number of extracted elements remains at most a constant proportion of the entire domain, these modifications do not significantly impact the algorithm, although we will handle them carefully.

Following~\autoref{lemma:collisions}, a vertex of size $R$ contains $\bigO{}(R^2 / 2^m)$ multicollisions on average, and the standard deviation is $\bigO{}(R/2^{m/2})$. This means that if we extract $\bigO{}(R/2^{m/2})$ multicollisions and measure them, we collapse on a superposition of vertices which are still \emph{typical}: they only contain slightly less multicollisions than average. Using~\autoref{cor:flip}, we will show that we can always change a \emph{typical} superposition of vertices into another typical one after $\bigO{1}$ steps of quantum search, and in our case, of quantum walk.
We can thus transform our superposition of vertices into vertices having slightly \emph{more} multicollisions than average, repeat the extraction, and so on. Therefore, at the cost of a total time $\bigOt{}(\max(\sqrt{R}, R/2^{m/2}))$, we can extract $\bigO{}(R/2^{m/2})$ multicollisions. As we have $\sqrt{R} < 2^{m/2} \implies R / 2^{m/2} < \sqrt{R}$, the time is dominated by $\bigOt{}(\sqrt{R})$.

We find that $2^k$ is always larger than $(R/2^{m/2})$ and the total complexity is:
$$ R + \frac{2^k}{(R/2^{m/2})} \times \sqrt{R} = R +  \frac{2^k 2^{m/2}}{\sqrt{R}} \enspace. $$
Optimizing $R$, we recover the time $\bigOt{ 2^{2k/3 + m/3} }$.
% 2^{2k/3 + m/3 - m/2} = 2k/3 - m/6 < k

This new regime, together with the previous one of~\cite{DBLP:conf/eurocrypt/BonnetainCSS23} (where the vertex contains less than one multicollision on average) covers all possible parameters for $k$ and $m$. It also completes the time-memory trade-off curve, where for any $\ell \leq 2k/3 + m/3$, there exists an algorithm with memory $\bigOt{2^\ell}$ and time $\bigOt{ 2^{k + m/2 - \ell/2} }$: for small values of $\ell$ the algorithm is the one of~\cite{DBLP:conf/eurocrypt/BonnetainCSS23}, and for large values of $\ell$ it is ours.

%At $k = n = m$, the vertex is of size $2^n$. We extract collisions by chunks of $2^{n/2}$ (although we could also run a classical algorithm). At $k = 0$ and $n = m$, 

%
%Our main algorithm features a very specific use of MNRS-style quantum walks, in which we do a single quantum walk step. Thus it is the \emph{diffusion} (the reflection through the uniform distribution over the current Johnson graph) which is the most important.
%
%We consider the following parameter range: we are looking for $2^k$ distinct collision pairs in the random function $h~: \zo^n \to \zo^m$ such that $k \leq 2n-m$ (otherwise such collisions don't exist) and $k \geq \frac{m}{4}$. As before, we do a series of walks on Johnson graphs $J(X, R)$ where the set $X$ and the size $R$ vary from one walk to the next. If the current table of extracted (multi)-collisions is denoted $C$, $X$ excludes $\mathsf{Preim}(C)$ and $R$ is progressively reduced.
%
%%Given a certain $X$ and $R$, there exist a given function $\mathsf{Collisions}(X,R)$ such that:
%%\[ \mathrm{Pr}_{ S \subseteq X, |S| = R } \left[ \text{$S$ contains at least $\mathsf{Collisions}(X,R)$ multicollisions} \right] \]
%
%For a given random set $S \subseteq X$ of size $R$, we can expect $C(X, R) = \bigO{ R^2 / 2^m}$ (multi)-collision tuples. \AS{TODO: remains to prove} Because the number of collisions follows a Chernoff-like concentration, we can also expect that the standard deviation from $C(X,R)$ is $\bigO{\sqrt{C(X,R)}}$. \AS{TODO: to prove, and there'll be constants}

\section{New Algorithm}

This section is dedicated to the details of our new algorithm. We will rely mostly on:
\begin{itemize}
\item \autoref{cor:flip} to \emph{flip} a superposition of vertices into its complementary, using a constant number of walk steps, as long as the proportion of such vertices is constant;
\item \autoref{cor:flip} to \emph{correct} a superposition of \emph{atypical vertices} back into a superposition of typical ones;
\item \autoref{lemma:collisions} for the statistics on the number of collisions.
\end{itemize}

%Our new algorithm handles a parameterization of the chained quantum walk in which the current vertex contains, on average, more than one collision. In that case the walk step itself becomes trivial, and consists only in a diffusion step. Consider for example a set of $2^r$ random queries to the function $f$, such that $r \geq m/2$. It contains $\bigO{2^{2r-m}}$ collisions on average, and the standard deviation is $\bigO{2^{r- m/2}}$. This means that we can remove $\bigO{2^{r- m/2}}$ collisions and still obtain a state which is (somewhat) close to the average behavior. Then, by performing a single walk iterate, we go back to the random case, and we can continue this procedure until we have extracted all the necessary collisions.

%We set $E$ a bound (little below the average) such that with overwhelming probability a random vertex has at least $E$ collisions. We define a deviation amount $T$ such that with probability $1/2$, a vertex has at more than $E+T$ collisions. $T$ is the number of collisions that we can extract from the vertex before the next diffusion step.
%
%Starting from a uniform vertex, we do the following: we extract $T$ collisions. We obtain a vertex with less collisions. We then do a diffusion step to make it uniform again. Each step of the chain takes time $\bigOt{T + 2^{r/2}}$.

% This is different from the previous algorithm because we cannot externalize easily the success probability. 

\subsection{Preliminaries: Restricted Functions}

%The exceptional case (a function having a larger multicollision than expected by this bound) is taken as a case of failure of our algorithm. From now on, we consider a function $f~: X \to Y$ where $|X| = 2^n$ and $|Y| = 2^m$ and their elements are easily described.

Throughout this section we write $X = \{0,1\}^n$ and $Y = \{0,1\}^m$. 

The multicollision tuples returned by our algorithm are stored in a classical, quantum-accessible table $C$ with entries of the form: $u : \{x_1, \ldots, x_r\}$ where $f(x_1) = \ldots = f(x_r) = u$. The set of \emph{images} of $C$ is defined as: $I_C = \{u \in Y, \exists t \subseteq X \text{ s.t. }(u:t)\in C\}$. We have $I_C \subseteq Y$. 
%\YS{ce n'est pas clair que t est un tuple ici.} \JL{Je suis d'accord. On pourrait réécrire par exemple $I_C = \{ u \in Y , \exists t \in X^r \text{ s.t. }(u:t)\in C \} \subseteq Y$} 
The set of \emph{preimages} of $C$ is defined as: $P_C = \{x \in X, \exists u \in I_C, f(x) = u\} \subseteq X$. We caution the reader here that $P_C$ is in general bigger than the multicollision tuples stored in $C$, because $C$ does not necessarily contain all preimages.

We can ensure that throughout the algorithm, $|P_C| < 2^{n-1}$. Indeed, even in the corner case where $n = m$ and we want $\bigO{2^n}$ collisions, we can stop the algorithm after returning $2^n/\bigO{n}$ multicollision tuples, and repeat at most $\bigO{n}$ times afterwards (ultimately, a polynomial factor in the time complexity).

Given a table $C$, we define a restricted function $f_C$ that excludes all preimages of $C$, and all corresponding images:
\begin{equation}
\begin{cases}
f_C ~: X \backslash P_C \mapsto Y \backslash I_C \\
f_C(x) = f(x) \enspace.
\end{cases}
\end{equation}

%Multicollisions of $f_C$ are obviously also multicollisions of $f$.

\paragraph{Vertices and Data Structure.}
Given the table $C$, given a size parameter $R$ (which varies during our algorithm), we define a set of vertices for one of our walks:
\begin{equation}
V^{R,C} := \{ S \subseteq (X \backslash P_C), |S| = R \} 
\end{equation}
and furthermore, we define $V^{R,C}_{x,y} \subseteq V^{R,C} $ as the subset of vertices \emph{containing between $x$ and $y$ multicollisions}. We authorize $y = \infty$ in this notation, where $V^{R,C}_{x,\infty}$ means all vertices containing more than $x$ multicollisions. By definition of $P_C$, the corresponding images of these multicollisions are ensured to not appear in $I_C$.

Like in previous work~\cite{DBLP:conf/eurocrypt/BonnetainCSS23}, subsets of $X$ are identified with (orthogonal) quantum states thanks to the quantum radix tree data structure~\cite{DBLP:phd/basesearch/Jeffery14}, and we can operate efficiently on them. Therefore, we can introduce the notation $\ket{V^{R,C}_{x,y}} := \sum_{S \in V^{R,C}_{x,y}} \ket{S}$ as a uniform superposition of such sets. The set data structure also encodes additional data:
\begin{itemize}
\item The outputs of $f$;
\item The number of multicollisions and their locations. This can be added to the radix tree by storing at each node the number of multicollisions that occur in its sub-tree.
\end{itemize}

Polynomial-time algorithms exist for:
\begin{itemize}
\item Inserting an element into the set;
\item Checking if an element belongs to the set;
\item Creating a uniform superposition of (multi)collisions in the vertex.
\end{itemize}

\subsection{Superposition of Elements}

%\AS{TODO: notation $X = \{0,1\}^n, Y = \{0,1\}^m$, to keep or not?}

%While the algorithm runs, we are actually finding (multi)collisions in the function $f_C$:
%\begin{equation}
%f_C~: X \backslash P_C \to Y \backslash I_C \enspace.
%\end{equation}
%
%Since we have ensured $|P_C| < |X|/2$, the set $P_C$ is not ``too big'' and so the behavior of $f_C$ is close to the one of $f$. That is, the number of (multi)collisions, and distribution thereof, are similar up to a constant factor.

%restrict the function $f$ to a function $f_C$, whose definition depends on our current table $C$ of results. First of all, 
%%Mwe can show that the restricted function $f_C$ has a similar number of collisions as $f$.
%
%\begin{lemma}
%\AS{TODO}: as long as $I_C$ is not too big (so $P_C$ as well), $f_C$ has approximately the same number of collisions as $f$ (and a similar distribution on the number of collisions)
%\end{lemma}

During our algorithm, we populate the table $C$ of multicollisions. However we cannot compute $P_C$, the set of all their preimages (we only know the preimages which belong to the table $C$). Despite this limitation, we can still create a uniform superposition over the set $X \backslash P_C$: intuitively, we just have to perform rejection sampling.

\begin{lemma}
Given $C$ stored in quantum-accessible memory such that $|C| \leq 2^{n-1}$, there exists a unitary to construct a uniform superposition over $X \backslash P_C$:
\begin{equation}
\ket{0} \mapsto \sum_{x \in (X \backslash P_C)} \ket{x} \enspace.
\end{equation}
with negligible error and polynomial time.
\end{lemma}

\begin{proof}
The idea is to perform a (quantum) rejection sampling, starting from:
\[ \sum_{x \in X} \ket{x} \]
and testing whether $x \in P_C$ efficiently thanks to the quantum RAM access. We obtain:
\[ \sqrt{\frac{|X| - |P_C|}{|X|}} \sum_{x \in (X \backslash P_C)} \ket{x} \ket{0} + \sqrt{\frac{|P_C|}{|X|}} \sum_{x \in P_C} \ket{x} \ket{1} \enspace. \]
We amplify the part 0 using robust amplitude amplification techniques. In particular, while we do not know the size of $P_C$, an upper bound of $|P_C| < |X| / 2$ is sufficient to succeed in polynomial time. \qed
\end{proof}

By using the same algorithm with an additional set of forbidden elements, we obtain a polynomial-time algorithm that creates a uniform superposition of neighboring vertices, which is sufficient to implement the \textbf{Update} unitary of our walks.

\subsection{Algorithm}

We consider quantum walks on Johnson graphs. The set of vertices for each walk is given by $V^{R,C}$ as defined above, for $C$ the current table of multicollisions, and $R$ the current vertex size. As explained above, $R$ remains similar to the initial vertex size.

Our algorithm (\autoref{alg:collision}) alternates between \emph{walk steps} and \emph{extraction steps}, which are detailed later in this section. An extraction step starts from a superposition of typical vertices with more collisions than average, $\ket{V^{R,C}_{E,E+T}}$, and produces a superposition $\ket{V^{R,C}_{E-T,E}}$ with less collisions. A walk step starts from $\ket{V^{R,C}_{E-T,E}}$ and restores a superposition of vertices $\ket{V^{R,C}_{E,E+T}}$ suitable for the next extraction step.

\begin{algorithm}[htbp]
\caption{Multiple collision-finding algorithm.}\label{alg:collision}

\begin{algorithmic}[1]
\Statex Choose an initial vertex size $2^{\ell}$
\medskip
\State Initialize $R \leftarrow 2^\ell$
\State Initialize $C \leftarrow \emptyset$
\State Initialize the state $\ket{V^{R,C}}$ (uniform superposition of vertices)
\State Project on $\ket{V^{R,C}_{E,E+T}}$ using~\autoref{cor:flip}
\medskip
\For{ $2^k 2^{m/2} / R$ iterations }\label{step:for}
\Statex \textbf{Extraction step:}
\State \textbf{Extract} $T$ multicollision tuples and measure them
\State The state becomes: $\ket{V^{R',C'}_{E-T,E}}$ for a new $R', C'$
\State \textbf{Extract} a little more to correct the state into $\ket{V^{R',C'}_{E'-T,E'}}$
\Statex \Comment{Here $E'$ is the new expectation value, which depends on $R'$, and is slightly smaller than $E$}
\State Update $R$ and $C$
\medskip
\Statex \textbf{Walk step.} From now on we write only the intervals on the number of collisions. We start with the internal $[E'-T; E']$ and we want to obtain $[E'; E'+T]$.
%\State Transform the state into $[0; E'-T'] \cup [E'; +\infty]$ using~\autoref{cor:flip}
% $\ket{V^{R',C'}_{0, E'-T'}} + \ket{V^{R',C'}_{E', \infty}}$ using~\autoref{cor:flip}
\State Measure the interval on the number of collisions \label{step:case}
\Statex \Comment{Possibilities: $[0; E'-T']$, $[E'-T'; E']$; $[E'; E'+T']$, $[E'+T'; \infty]$, each having a constant probability}
\State \textbf{Case} $[0; E'-T']$: transform the state into $[E'-T'; \infty]$, go to Step~\ref{step:case}
\State \textbf{Case} $[E'-T; E']$: transform the state into $[0; E'-T'] \cup [E'; +\infty]$, go to Step~\ref{step:case}
\State \textbf{Case} $[E'; E'+T]$: go to Step~\ref{step:for}
\State \textbf{Case} $[E'+T; \infty]$: transform the state into $[0;E'+T']$, go to Step~\ref{step:case}
\EndFor
\end{algorithmic}

\end{algorithm}

We start by analyzing the walk steps. By~\autoref{lemma:collisions} and~\autoref{lemma:modif-e}, we know that the number of multicollisions in vertices does not vary much; its average depends on $R$, but only slightly decreases at each loop iteration. We can adapt the number of extracted multicollisions to correct the interval appropriately (we will see later that multicollisions are extracted one by one).

After obtaining the state $\ket{V^{R',C'}_{E'-T,E'}}$, we transform it into $\ket{V^{R',C'}_{E',E'+T}}$ using a series of applications of~\autoref{cor:flip}. For each application, we define a different partition of ``good'' and ``bad'' vertices depending on the number of multicollisions that they contain. We have defined four intervals which form a partition of all vertices for the walk. As long as we do not obtain the wanted interval $[E';E'+T]$, applying~\autoref{cor:flip} and projecting transforms the current interval into another one of the three, where each of them has a constant probability to appear. Consequently this process terminates with an expected constant number of steps. Each step contains an expected constant number of single-step quantum walks.

\subsection{Extraction Steps}

We detail the \emph{extraction} routine which transforms $\ket{V^{R,C}_{E, E+T}}$ into $\ket{V^{R',C'}_{E-T, E}}$ for modified parameters $R',C'$, and returns $T$ multicollision tuples (classical values). Actually, we give a more general routine to extract a single multicollision tuple from a state $\ket{V^{R,C}_{x,y}}$, where $x,y$ are input parameters.

%\begin{remark}
%Notice that $E$ is still the average relative to the former parameters $R$ and $C$. However, as shown in~\autoref{lemma:modif-e}, the modification of $R,C$ will be small enough so that the effect on $E$ is negligible, and the state $\ket{V^{R',C'}_{E-T, E}}$ is close enough to $\ket{V^{R',C'}_{E'-T, E'}}$ to start the subsequent steps.
%\end{remark}

% a superposition of $\ket{V^{R,C}_{x,y}}$ into another one $\ket{V^{R',C'}_{x',y'}}$ for some $R', C', x', y'$, and a number of multicollision tuples (classical values). 

\paragraph{Basic Idea.}
Given a vertex $V$, we know thanks to the tree data structure how many multicollisions appear and where they are located. The basic idea of the extraction routine is to create a uniform superposition of these collisions, outside of the tree, into a separate register, and to measure:
\begin{itemize}
\item the size of the register, which projects on multicollisions having the same number of elements (and determines the new parameter $R'$);
\item the actual elements of the multicollisions, which gives the new $C'$.
\end{itemize}

This transforms the state $\ket{V^{R,C}_{x,y}}$ into a superposition $\ket{V^{R',C'}_{x',y'}}$ of vertices which exclude any preimage of the newly measured multicollision, and contain between $x-1$ and $y-1$ collisions. 

\paragraph{An Issue with post-Selection.}
There is a subtlety due to the fact that the vertices in $\ket{V^{R,C}_{x,y}}$ do not contain the same number of multicollisions. Consider the following example of a superposition with a vertex having 3 collisions $\ket{c,c', c''}$ and a vertex having two collisions $\ket{c,c'}$. Let us simplify the writing of these vertices by considering only the collisions, furthermore we write $\ket{\{c,c'\}}$ for an unordered \emph{set} of elements. We will select one of the collisions and move it to a fixed position in the quantum memory, which we measure. Before measurement, the quantum state has the form:
\begin{equation}
\frac{1}{\sqrt{3}} \ket{c}\ket{\{c',c''\}} + \frac{1}{\sqrt{3}} \ket{c'}\ket{\{c,c''\}} + \frac{1}{\sqrt{3}} \ket{c''}\ket{\{c,c'\}} + \frac{1}{\sqrt{2}} \ket{c}\ket{\{c'\}} + \frac{1}{\sqrt{2}} \ket{c'}\ket{\{c\}} \enspace.
\end{equation}
Now, we measure. Suppose that we obtain $c$, then the state becomes proportional to the following, non-uniform superposition:
\begin{equation}
\frac{1}{\sqrt{3}} \ket{\{c',c''\}} + \frac{1}{\sqrt{2}} \ket{\{c'\}} \enspace.
\end{equation}

%We permute randomly the collisions, moving one of them to a fixed position, which we measure: 
%\YS{pourquoi seulement la permutation cyclique ici? Il y a aussi un typo sur les primes sur c.}\AS{Je modifie l'écriture de cette ligne. En gros, on pousse la collision à une position fixée, mais ce qui reste c'est un *ensemble* et non une *séquence* de trucs.}
%\[ \frac{1}{\sqrt{3}} \ket{c,c',c''} + \frac{1}{\sqrt{3}} \ket{c',c,c'} + \frac{1}{\sqrt{3}} \ket{c'',c',c'} + \frac{1}{\sqrt{2}} \ket{c,c'} + \frac{1}{\sqrt{2}} \ket{c',c} \]
%measuring $c$ then collapses the state on a non-uniform superposition:
%\[ \frac{1}{\sqrt{3}} \ket{c', c''} + \frac{1}{\sqrt{2}} \ket{c'} \enspace. \]

The issue here is that there is a post-selection depending on the number of collisions in the vertices, which is not constant in our algorithm. As we need the superposition to remain uniform, we need to correct this deviation.

\paragraph{Procedure.}
Let us start from the superposition $\ket{V_{x,y}^{R,C}}$. In a given vertex $V$ with $z$ multicollisions, let $(t_i)_{1 \leq i \leq z}$ be the multicollision tuples. When we create the uniform superposition of multicollisions, we obtain the state: 
%\YS{ça ne correspond pas à l'example ci-dessus?}\AS{maintenant si?}
\begin{equation}
 \sum_{x \leq z \leq y} \sum_{V \in V^{R,C}_{z}} \ket{V} \left( \frac{1}{\sqrt{z}} \sum_{1 \leq i \leq z} \ket{ t_i } \right) \enspace. 
 \end{equation}
The problem here is that the amplitude over a given $t_i$ depends on $z$ (the number of multicollisions in the vertex). So we make it constant by adding new ``dummy'' values $d_i$:
\begin{equation}
 \sum_{x \leq z \leq y} \sum_{V \in V^{R,C}_{z}} \ket{V} \left(  \frac{1}{\sqrt{y}} \sum_{1 \leq i \leq z} \ket{ t_i } + \frac{1}{\sqrt{y}} \sum_{z+1 \leq i \leq y} \ket{d_i} \right) \enspace.
 \end{equation}
Now, we measure the last register:
\begin{itemize}
\item With probability $ \frac{1}{|V^{R,C}_{x,y}|} \sum_{x \leq z \leq y} |V^{R,C}_{z}| \frac{z}{y} $, we obtain a multicollision. We then collapse on a uniform superposition of vertices $V^{R',C'}_{x-1,y-1}$ where $R'$ is reduced by the size of the multicollision, and $C'$ contains this new multicollision (as a result of post-selection, the vertex excludes any preimage of the new value). This is what we wanted.

\item For all $x + 1 \leq i \leq y$, we measure $d_i$ with probability:
$$ \frac{1}{|V^{R,C}_{x,y}|} \sum_{x \leq z \leq i-1}  |V_{z}^{R,C}| \frac{1}{y} = \frac{1}{y} \frac{|V_{x, i-1}^{R,C}|}{|V^{R,C}_{x,y}|} $$
In that case we project the vertex on:
\begin{equation}
 \sum_{x \leq z \leq y | z+1 \leq i} \sum_{V \in V^{R,C}_{z}} \ket{V} = \ket{V^{R,C}_{x, i-1}} \enspace.
\end{equation}
Indeed, we still have a uniform superposition of sets, and a set belongs in this superposition if and only if it has less than $i-1$ multicollisions (in which case, the dummy $\ket{d_i}$ is created). 
% $V_{i, y}^{R,C}$. 
% \YS{$V_{x, i}^{R,C}$? Si c'est ça, il faut modifier la preuve du lemme 7.}
There are no new multicollisions extracted; we have simply changed the bounds on the number of multicollisions in the superposition of vertices.
\end{itemize}

The latter case is problematic for us, as we have refined the superposition: this makes the vertex possibly more ``atypical''. Fortunately, we can always go back to the typical case, at a cost depending on the exact proportion of vertices that we fell upon, using~\autoref{lemma:qaa-bad}.

\begin{lemma}
Assume that $V^{R,C}_{x,y}$, $V^{R,C}_{0,x}$, $V^{R,C}_{y,\infty}$ contain a constant proportion of all vertices $V^{R,C}$, and $x/y = \bigO{1}$. Then there is a quantum algorithm with \emph{average time} $\bigO{ \frac{\sqrt{R} (y-x)}{y} }$ that, on input $\ket{V^{R,C}_{x,y}}$ extracts \emph{on average} a constant number of multicollisions and returns an updated $\ket{V^{R',C'}_{x-1,y-1}}$.
\end{lemma}

\begin{proof}
The probability to obtain a multicollision immediately is bigger than $x/y$, which is constant. If we do not obtain a multicollision, we need to revert back to $\ket{V^{R,C}_{x,y}}$.

For all $i \geq x$, we project on $\ket{V_{x, i}^{R,C}}$ with probability $\frac{1}{y} \frac{|V_{x, i}^{R,C}|}{|V^{R,C}_{x,y}|}$. Via a series of a constant number (on expectation) of manipulations, we will transform the interval $[x; i]$ into $[0;x[$, then $[x;y]$.
\begin{itemize}
\item Starting from $[x;i]$, we use~\autoref{cor:flip} to obtain $[0;x] \cup [i; \infty[$ using an expected number of $\bigO{\sqrt{\frac{|V^{R,C}|}{|V^{R,C}_{x,i}|}  } }$ iterations. With constant probability we project on $[0;x]$ or $[i;\infty[$.
%\item Starting from $[0;i]$, we use~\autoref{cor:flip} to obtain $[i;\infty[$ with $\bigO{1}$ iterations. With constant probability we project on $[i;y]$ or $[y;\infty[$. (If $i$ is close to $y$, the probability to project on $[i;y]$ is actually low).
\item Starting from $[i; \infty[$, we use~\autoref{cor:flip} to obtain $[0;i[$ with $\bigO{1}$ iterations. With constant probability we project on $[0;x]$ or $[x;i]$. (If $i$ is close to $x$, the probability to project on $[x;i]$ is actually low). In the latter case, we go back to the previous step.
\end{itemize}
%Starting from $[y;\infty[$ we use~\autoref{cor:flip} to obtain $[0;x] \cup [x;y]$. We project on $[0;x]$ with constant probability, and in that case we go to $[x;y] \cup [y;\infty[$, which either brings us back to $[y;\infty[$, or succeeds.

As a consequence, for all $i \geq x$, we succeed later with $\bigO{\sqrt{\frac{|V^{R,C}|}{|V^{R,C}_{x,i}|}  } }$ iterations on average. By taking the average over all $i$, the number of iterations that we need is, up to a constant:
\begin{equation}
\sum_{i = x}^y \frac{1}{y} \frac{|V_{x,i}^{R,C}|}{|V^{R,C}_{x,y}|}  \sqrt{\frac{|V^{R,C}|}{|V^{R,C}_{x,i}|}  }  = \bigO{ \frac{1}{y} \sum_{i = x}^y \sqrt{\frac{|V^{R,C}_{x,i}|}{|V^{R,C}_{x,y}|} } } = \bigO{ \frac{y-x}{y} } \enspace.
\end{equation}
Here we use the fact that $|V^{R,C}|$ and $|V^{R,C}_{x,y}|$ are the same up a constant, and then, that $|V^{R,C}_{x,i}| \leq |V^{R,C}_{x,y}|$. Each iteration costs $\sqrt{R}$ queries and $\bigOt{\sqrt{R}}$ time, giving the result.\qed
\end{proof}

We use this algorithm for $\ket{V^{R,C}_{E, E+T}}$ where $x = E$ and $y = E+T$, which satisfy its conditions. The average time is $\bigO{ \sqrt{R} T / E }$, where $T = \bigO{ R/\sqrt{M} }$ and $E = \bigO{R^2/M}$. Since we need to extract $T$ multicollisions, the total average time is $\bigO{ \sqrt{R} T^2 / E} = \bigO{ \sqrt{R}}$, the same as the walk step in~\autoref{alg:collision}.

\section{Conclusion}
In this paper, we proved that the multiple collision search problem quantum complexity in queries (and time) is $2^{2k/3 + m/3}$ where $m$ is the output bit-size and $2^k$ the number of collisions, closing the remaining gap in some range of $k$ and $m$. In the new targeted regime, we use MNRS quantum walks as a tool for \emph{diffusion} in the current state rather than \emph{search}, which allows to correct the expected number of collisions in this state -- allowing to extract collisions in subsequent steps. This peculiar use of the diffusion operator in the MNRS walk could perhaps have further applications.

% %Y.S. is supported by EPSRC grant EP/S02087X/1 and EP/W02778X/1. 

%\ifeprint
%\subsubsection*{Acknowledgments.}
%The authors would like to thank Schloss Dagstuhl and the organizers of the Dagstuhl Seminar 23421 ``Quantum cryptanalysis'' where this project was initiated. 
%This work has been partially supported by the French Agence Nationale de la Recherche through the France 2030 program under grant agreement No. ANR-22-PETQ-0008 PQ-TLS.
%\fi

%\nocite{*}
\bibliography{biblio}
\bibliographystyle{splncs03}
%\printbibliography

%\newpage
%\appendix

%\input{./files/appendix.tex}

\end{document}